\documentclass[10pt,conference,letterpaper]{IEEEtran}

\usepackage{color}
\usepackage{epsfig}
\usepackage{graphicx}
\usepackage{epic}
\usepackage{amsfonts}
\usepackage{latexsym}
\usepackage{amsmath}
\usepackage{amsthm}
\usepackage{amssymb}
\usepackage{epstopdf}

\newcommand{\suppressfig}[1]{{#1}}

\newcommand{\isitout}[1]{{#1}}
\newcommand{\isitin}[1]{}

\newtheorem{theorem}{Theorem}%[section]
\newtheorem{question}{Question}%[section]
\newtheorem{lemma}{Lemma}%[section]
\newtheorem{claim}{Claim}%[section]
\newtheorem{proposition}{Proposition}%[section]
\newtheorem{statement}{Statement}%[section]
\newtheorem{definition}{Definition}%[section]
\newtheorem{corollary}{Corollary}%[section]
\newcommand{\eps}{{\varepsilon}}

\newcommand{\x}{{x^{(n)}}}

\newcommand{\y}{{y^{(n)}}}

\newcommand{\Reps}{{\mathcal{R}^{\left(\eps\right)}_{W}}}
\newcommand{\Repsn}{{\mathcal{R}^{\left(\eps\right)}_{W,n}}}
\newcommand{\Rzero}{{\mathcal{R}^{\left(0\right)}_{W}}}
\newcommand{\Rzeron}{{\mathcal{R}^{\left(0\right)}_{W,n}}}

\def\cW{\mbox{$\cal{W}$}}
\def\cX{\mbox{$\cal{X}$}}

\def\cN{\mbox{$\cal{N}$}}

\def\cR{\mbox{$\cal{R}$}}

\def\e{\varepsilon}

\newcommand{\xn}{x^{(n)}}
\newcommand{\yn}{y^{(n)}}

\newcommand{\xnh}{{\hat{x}^{(n)}}}
\newcommand{\ynh}{{\hat{y}^{(n)}}}

\def\01{\{0,1\}}

\newcommand{\remove}[1]{}

\begin{document}

\IEEEoverridecommandlockouts

\title{Zero vs. $\e$ Error in Interference Channels}

\author{
\IEEEauthorblockN{I. Levi}
%\IEEEauthorblockA{Chinese University of Hong Kong}
\IEEEauthorblockA{Open University of Israel\\
{\em ilia.levi@gmail.com}}
\and
\IEEEauthorblockN{D. Vilenchik }
%\IEEEauthorblockA{Chinese University of Hong Kong}
\IEEEauthorblockA{Weizmann Institute of Science\\
{\em dan.vilenchik@weizmann.ac.il}}
\and
\IEEEauthorblockN{M. Langberg}
\IEEEauthorblockA{Open University of Israel\\
{\em mikel@openu.ac.il}}
\and
\IEEEauthorblockN{M. Effros}
%\IEEEauthorblockA{Chinese University of Hong Kong}
\IEEEauthorblockA{Caltech\\
{\em effros@caltech.edu}}
\thanks{The work of Michael Langberg was supported in part by ISF grant 480/08, BSF grant 2010075, and NSF grant 1038578. Work done in part while Dan Vilenchik was at The Open University of Israel and Michael Langberg was at the California Institute of Technology.
}
}

\maketitle

\begin{abstract}
Traditional studies of multi-source, multi-terminal interference channels typically allow a vanishing probability of error in communication.
Motivated by the study of network coding, this work addresses the task of quantifying the loss in rate when insisting on {\em zero} error communication in the context of interference channels.
\end{abstract}

\section{Introduction}
\label{sec:intro}
In the distributed multi-source/multi-terminal network coding paradigm, independent sources wish to convey their information to a set of terminals over a given network $\cN$ via a communication scheme in which internal nodes of the network may mix (i.e., encode) the information content of received packets before forwarding them (see e.g., \cite{ACLY00,LYC03,KM03,JSCEEJT04,HoMKKESL06} and references therein).
In such a communication scheme, each terminal eventually receives a certain function of the source information and is required to decode based on the information received. For example, in the multiple-unicast scenario, there are $k$ source/terminal pairs and terminal $i$ is required to decode the information of source $i$.

One may abstractly model the end-to-end behavior of a given multiple-unicast communication scheme by a corresponding $k$-source/$k$-terminal interference channel $W: \cX^k \rightarrow \widehat{\cX}^k$. Such a channel receives as input the encoded information $x=x_1,\dots,x_k \in \cX^k$ from the $k$ independent sources and returns as output a vector $\hat{x}=\hat{x}_1,\dots,\hat{x}_k \in \widehat{\cX}^k$, where $\hat{x}_i$ is the information available at terminal node $i$.
As an example, consider the famous butterfly network in Figure~\ref{fig:butterfly}.
The channel $W$, corresponding to the well known encoding scheme presented in the figure, sets $W(x_1,x_2)=(\hat{x}_1,\hat{x}_2)$ with $\hat{x}_1 = (x_2,x_1+x_2)$ and $\hat{x}_2 = (x_1,x_1+x_2)$.

As in the butterfly example, it is common in the network coding literature to assume that the corresponding channel $W$ is {\em deterministic} (i.e., it is completely determined by the source information) and that communication is considered successful if all terminals are able to decode the information they received, no matter what source information was transmitted. We refer to the latter requirement as {\em zero error communication}.

The question whether zero error communication poses a restriction on the achievable rate has seen recent interest \cite{CG10, LE11} and has been found in \cite{LE12,EEL12} to be closely related to additional intriguing questions such as the {\em edge-removal} problem \cite{HEJ:10,JEH:11}. Relaxing the requirement of zero error communication to that of $\e>0$ error (in which one allows communication to fail with probability $\e$ over the source messages) yields the following open question \cite{CG10,LE11}.
\footnote{We note that several statements below are made informally. Formal definitions and statements follow in Section~\ref{sec:model}.}

\begin{question}
\label{q:nc}
Let $\e>0$ be an arbitrarily small constant.
In the network coding paradigm, can one obtain a strictly higher rate of communication when allowing $\e$ error in communication as opposed to zero error?
\end{question}

To better understand the price in rate of the zero-error constraint in the context of network coding, in this work we study a relaxed version of Question~\ref{q:nc}.
Specifically, we view communication via network coding as communication over deterministic interference channels and study the potential gap in rate when communicating with zero error over deterministic interference channels as opposed to $\e >0$ error.

\suppressfig{
\begin{figure}[t]
  \begin{center}
\includegraphics*[viewport=118 130 326 368, scale=0.55]{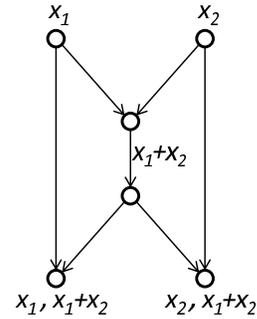}
\caption{The interference channel corresponding to the butterfly network takes input $x_1$ and $x_2$ and returns $(x_2, x_1+x_2)$ to the right terminal (which requires $x_1$) and $(x_1, x_1+x_2)$ to the left terminal (which requires $x_2$).}
    \label{fig:butterfly}
   \end{center}
\end{figure}}

\begin{question}
\label{q:ic}
Let $\e>0$ be an arbitrarily small constant.
Do there exist deterministic interference channels for which one can obtain a strictly higher rate of communication with $\e$ error as opposed to zero error?
\end{question}

A negative answer to Question~\ref{q:ic} would imply a negative answer to Question~\ref{q:nc}.
Resolving Q.\ref{q:ic} however does not necessarily resolve Q.\ref{q:nc}, since the channels that may affirmatively answer Q.\ref{q:ic}, could possibly not correspond to any given network coding topology.
Also, Q.\ref{q:ic} fixes a single network code on a given network topology, but it is not sufficient to study a single network code to resolve Q.\ref{q:nc} (as the coding scheme that achieves $\e$ error may differ from the best zero error scheme).

%
%The latter follows from two main reasons. Primarily, interference channels that result from network communication schemes must take into account the given topology of the network, while the channels that may answer Q.\ref{q:ic} in the affirmative may not correspond to any given topology. Secondly, to resolve Q.\ref{q:nc} one must take into account all possible network coding schemes corresponding to a given topology and collection of communication requirements (as the coding scheme that achieves $\e$ error may differ significantly from the best zero error scheme), whereas in Q.\ref{q:ic} we fix the coding scheme (i.e., interference channel) under study.

The answers to Q.\ref{q:nc} and Q.\ref{q:ic} are known to be positive when information
transmitted from different sources is {\em dependent}. That is, allowing an $\e$-error can significantly
increase the achievable rate region, as shown, for example, for the Slepian-Wolf problem in \cite{SlepianW:73}.
In the network coding setting, we assume that sources are independent.

\subsection{Our contribution}

The main focus of this work is to better understand Q.\ref{q:ic} and, in light of its connections with Q.\ref{q:nc}, to gain a better understanding of the tradeoff between $\e>0$ and zero error in network coding.

Our work focuses on the 2-source/2-terminal setting. While not resolving Q.\ref{q:ic}, we present and analyze a family of deterministic interference channels $\cW$, which we believe can act as witnesses to an affirmative answer of Q.\ref{q:ic}, with arbitrarily small values of $\e>0$.

In Sections~\ref{sec:model} and \ref{sec:pre}, we present our channel model in detail and define a refined version of Q.\ref{q:ic} alongside preliminary results and previous work.
In Section~\ref{sec:upper}, we analyze the family $\cW$ discussed above and present a positive answer to Q.\ref{q:ic} assuming a finite communication blocklength $n$.
In Section~\ref{sec:lower}, we study what we view as a natural approach to refute Q.\ref{q:ic}, and show that it does not necessarily succeed. Finally we conclude in Section~\ref{sec:conc}.

\section{Model}
\label{sec:model}

In a multiple unicast communication network, the objective is for $k$ source nodes, $s_1, s_2, \ldots,  s_{k}$, to communicate their information to $k$ corresponding terminal nodes, $t_1, t_2, \ldots,  t_{k}$ over a channel $W$.
In this work, we focus on the case of two sources and two terminals (i.e., $k=2$). A discussion regarding our model and results for larger values of $k$ appears in Section~\ref{sec:conc}.
One can model a deterministic multiple unicast communication network with blocklength $n$ by the following components.
The model presented here differs slightly in notation from that presented informally in the Introduction; namely, to simplify notation for $k=2$, encoded source information is denoted by the pair $(x,y)$ and not $(x_1,x_2)$.

\textbf{Message space}: For $i=1,2$, source $s_i$ holds a message from a set of size $M_i$. Without loss of generality, the message space can be defined as $\left[ M_i \right] = \left\{1, \ldots, M_i\right\}$.

\textbf{Encoding}: For alphabet $[Q]=[2^q]$,\footnote{For simplicity, we assume that $Q$ is an integer power of two; our results hold for any $Q \geq 2$.} and block length $n$, each source $s_i$ holds an encoding function $E_i : \left[ M_i \right] \to [Q]^n$. We denote the coded information corresponding to source $s_1$ by $x^{(n)}=(x_1,\dots,x_n) \in [Q]^n$, and that corresponding to $s_2$ by $y^{(n)}=(y_1,\dots,y_n) \in [Q]^n$.

\textbf{Network $W$}: The network $W: [Q]^2 \rightarrow [Q]^2$ is a deterministic function that takes as input elements from $[Q]^2$ and returns elements from the same alphabet.
Denoting $W$ as $(W_1,W_2)$, terminal $t_i$ receives the evaluation of $W_i : [Q]^2 \to [Q]$ on input $(x,y)\in [Q]\times[Q]$.

\textbf{Network $W^{(n)}$}:
Applying the network $W$ $n$ times (for blocklength $n$) yields the network $W^{(n)} :  [Q]^n\times[Q]^n \to  [Q]^n\times[Q]^n$ which is a deterministic function that takes as input two $n$ vectors and returns two $n$ vectors. Namely, denoting $W^{(n)}$ as $(W^{(n)}_1,W^{(n)}_2)$, the evaluation of $W^{(n)}_1 :  [Q]^n\times[Q]^n \to  [Q]^n$ on input $(\xn,\yn) \in [Q]^n\times[Q]^n$ is a vector $\xnh \in [Q]^{n}$ received at terminal $t_1$, where $\xnh_j=W_1(\xn_j,\yn_j)$. Similarly, $W^{(n)}_2 (\xn,\yn)$ is a vector $\ynh \in [Q]^{n}$ received at terminal $t_2$, where $\ynh_j=W_2(\xn_j,\yn_j)$.

\textbf{Decoding}: Each terminal $t_i$ holds a decoding function $D_i : [Q]^{n} \to \left[ M_i \right]$.

Communication with block length $n$ is successful for terminal $t_i$ and source information $(m_1,m_2)$ if for $i=1,2$,
$D_i [W^{(n)}_i(E_1(m_1),E_2(m_2))] = m_i$. We say that communication is successful with probability $1-\e$ if for source information $(m_1,m_2)$ chosen uniformly at random from $[M_1]\times[M_2]$ it holds with probability $1-\e$ that communication is successful for all terminals.
%Namely,
%\begin{equation}
%\label{commsucc}
%\Pr_{m_1,\dots, m_r} \left[\text{communication is successful for all terminals}\right] \geq 1-\e
%\end{equation}
Rate $(R_1,R_2)$ is achievable with probability $1-\e$ and block length $n$ over network $W$ if for $M_i = 2^{R_i n}$ there exist encoding and decoding functions such that communication is successful with probability $1-\e$.

The $\e$-error sum capacity of a network $W$ and block length $n$ is defined to be
\begin{equation*}
\cR^{\left(\e\right)}_{W,n} = \sup_{(R_1,R_2) \in \Gamma_{n, \e}}(R_1+R_2),
\end{equation*}
where the supremum is taken over the set $\Gamma_{n, \e}$ of rate pairs $(R_1,R_2)$ that are achievable with probability $1-\e$ and block length $n$ over $W$.
The $\e$-error sum capacity of a network $W$ is defined as
\begin{equation*}
\cR^{\left(\e\right)}_{W} = \sup_n \cR^{\left(\e\right)}_{W,n}.
\end{equation*}
In particular, for $\e = 0$, we have $\cR^{\left(0\right)}_{W}$.
We here study the relationship between  $\cR^{\left(\e\right)}_{W}$ and $\cR^{\left(0\right)}_{W}$.

\isitout{Some remarks are in place.
Our model implies independence in encoding (i.e., sources cannot communicate with each other) and independence in decoding (i.e., terminals cannot communicate with each other), which is a commonly used and realistic model.
Also notice that $W$ can be defined probabilistically and not deterministically as above. We do not address probabilistic $W$ in this work, but one may prove that Q.\ref{q:ic} has a positive answer in this context.\footnote{For example, consider the channel $W$ which on input $(x_1,x_2)$ returns $(x_1,x_2)$ with probability $1-\e$ and a random pair $(x_1',x_2')$ chosen uniformly from $[Q]^2$ with probability $\e$.}}
\isitin{We note that $W$ can be defined probabilistically and not deterministically as above. We do not address probabilistic $W$ in this work, but one may show that Q.\ref{q:ic} has a positive answer in this context.\footnote{For example, consider the channel $W$ which on input $(x_1,x_2)$ returns $(x_1,x_2)$ with probability $1-\e$ and a random pair $(x_1',x_2')$ chosen uniformly from $[Q]^2$ with probability $\e$.}}

\section{Preliminaries and previous work}
\label{sec:pre}
Given a channel $W$, our main interest in this work is the relationship between  $\mathcal{R}^{\left(0\right)}_{W}$ and  $\mathcal{R}^{\left(\e\right)}_{W}$.
In words, $\mathcal{R}^{\left(0\right)}_W$ represents the achievable rate when communicating with no error at all, while $\mathcal{R}^{\left(\e\right)}_W$ represents the rate when allowing a small $\e$ probability of error. %There are several examples in communication in which allowing an $\e$-error significantly increases the rate of communication when compared to zero-error \cite[Chapter 15.4]{EoIT06s154}.
Specifically, we explore the plausibility of the following open statement which claims a large gap between $\mathcal{R}^{\left(\e\right)}_W$ and $\mathcal{R}^{\left(0\right)}_W$.
The statement below is a refined version of Q.\ref{q:ic} above.
\begin{statement}
\label{state:1}
Let $\e > 0$. There exists $\delta=\delta(\e) > 0$ that tends to 0 when $\e$ tends to 0 such that for every network $W$ it holds that
\begin{equation*}
\frac{\mathcal{R}^{\left(\e\right)}_{W}}{2} - \delta  \leq \mathcal{R}^{\left(0\right)}_{W}.
\end{equation*}
Moreover, for $\e>0$ and $\delta$ as above, there exists a network $W_\e$, such that
\begin{equation*}
\mathcal{R}^{\left(0\right)}_{W_\e} \leq \frac{\mathcal{R}^{\left(\e\right)}_{W_\e}}{2}+\delta.
\end{equation*}
\end{statement}
In other words, for certain networks $W$, requiring zero-error in communication may reduce the sum capacity by a factor of $2$ (or equivalently, allowing an $\e$ error may increase the sum capacity by a factor of $2$), and this $2$-factor is tight.

It is simple to obtain the first part of Statement~\ref{state:1} via a {\em time sharing} scheme.

\begin{lemma}
\label{lemma:time_sharing}
For any $n,\e>0$, and $\delta=-\frac{\log(1-\e)}{n}$, any channel $W$ satisfies
$
\frac{\Repsn}{2} - \delta   \leq \Rzeron.
$
Here, $\delta>0$ tends to $0$ as $\eps$ tends to 0 or $n$ to $\infty$.
\end{lemma}
\begin{proof}
Let without loss of generality $M_i=2^{R_in}$ and $\Repsn=R_1+R_2$.
Assume that $R_1 \ge R_2$.
By definition of $\cR^{\left(\e\right)}_{W,1}$, there exists an $m^{*} \in M_2$ and a subset $S \subseteq M_1$ of size at least $(1-\e)|M_1|$ such that
$D_1 [W^{(n)}_1(E_1(m),E_2(m^*))] = m$ and $D_2 [W^{(n)}_2(E_1(m),E_2(m^*))] = m^*$ for every $m \in S$. Taking $M_1' = S$ and $M_2'=\{m^*\}$, we get a zero-error communication scheme over $W_n$ with sum rate
$\frac{1}{n}\left(\log |M_1'|+\log |M_2'|\right)=R_1+\frac{\log(1-\e)}{n}$. Since we assumed $R_1 \ge R_2$, we get
\begin{eqnarray*}
\Rzeron
& \ge &
R_1+\frac{\log(1-\e)}{n} \ge \frac{\Repsn}{2}+\frac{\log(1-\e)}{n}.
\end{eqnarray*}
\end{proof}
As a corollary of Lemma \ref{lemma:time_sharing} we get the first part of Statement \ref{state:1}: Fix $\eps$, and let $\delta^*=-\log(1-\eps)$. Lemma~\ref{lemma:time_sharing} implies that for all $n$,
$$\sup_n \left (\frac{\Repsn}{2} - \delta^* \right)  \leq  \sup_n \Rzeron.$$ Since $\delta^*$ does not depend on $n$, we can take it out of the parentheses, giving the first part of Statement \ref{state:1}.
Recall the definitions $\sup_n \Rzeron = \Rzero$ and $\sup_n \Repsn=\Repsn$.

%In the proof above we communicated information to $t_1$ only, but we can similarly communicate to $t_2$ by finding the set $S_x^{*}$, for which the same bounds apply. The time sharing scheme is defined by alternating between these communications at each application of the network $W$ in a round-robin fashion.
%The following lemma establishes a connection between $\Reps$ and $\Repsn$, which will be useful in proving our results.
%We note that in general the rate $\Repsn$ may increase with $n$, but will not decrease significantly (for small $\e$) as proven, e.g., in \cite{LE11}.

\subsection{Previous work}

In an excellent survey, K\"{o}rner and Orlitsky \cite{KO98} discuss the problem under study, and describe a special case of a 2-user network in which $[Q]=[2] \equiv \left\{0, 1\right\}$ and $W=(W_1,W_2)$ with
\begin{eqnarray*}
W_1\left(x, y\right) = \max\left(x, y\right), \ \ \
W_2\left(x, y\right) = \min\left(x, y\right).
\end{eqnarray*}
The problem addressed in \cite{KO98}  is to find $\mathcal{R}^{\left(0\right)}_{W}$.
It is not hard to verify that $\mathcal{R}^{\left(1/4\right)}_{W}=2$.
The authors note that this problem has a combinatorial formulation, and that $\Rzero$ is conjectured (by \cite{S:89} and \cite{AS:94}) to be equal to $1$ (which matches Statement~\ref{state:1} for $\mathcal{R}^{\left(1/4\right)}_{W}=2$). However, the best known upper bound is $\Rzero \leq 1.2118$ \cite{HK:95}. A sum rate $1$ is easily achieved by using the network to transmit the information of one user only. For $n=1$, define $E_1\left(0\right) = 0, E_1\left(1\right) = 1, E_2\left(0\right) = 0, E_2\left(1\right) = 0$ and $D_1\left(x,y\right)=x$.
Using the time-sharing scheme suggested above, we can convey information to both users, one at a time, with sum rate $1$.
The above proves Claim~\ref{claim:pre}.
\begin{claim}[\cite{HK:95}]
\label{claim:pre}
There exists a binary channel $W$ such that for $\e=1/4$,
$\Rzero \leq 0.6059 \cdot \Reps.$
\end{claim}
Our work addresses the potential gap between $\Rzero$ and $\Reps$ for arbitrary values of $\e>0$.
\subsection{``Erasure/identity'' channels}
As we have seen, the first part of Statement \ref{state:1} is true. In this work, we explore the second part of that statement. We conjecture that it is correct, and provide evidence that supports this conjecture. To this end, we analyze the gap between $\Reps$ and $\Rzero$ on a family of channels $W$ for which $W: [Q]^2 \rightarrow \left([Q]^2\cup {(\phi,\phi)}\right)$ is either the identity function (i.e., $W(x,y)=(x,y)$) or $W$ returns an ``erasure value'' (i.e., for a new symbol $\phi \not\in [Q]$, $W(x,y)=(\phi,\phi)$). Notice that we change the model slightly by allowing our output alphabet to have an additional symbol.
We refer to such channels as {\em erasure/identity} channels.
More specifically, we consider a distribution over erasure/identity channels $W$, and study the properties of the resulting channels.
Our distribution is very natural and is parametrized by $\e$.
\begin{definition}
Let $\cW_{Q,\e}$ be the distribution over erasure/identity channels in which for every $(x,y) \in [Q]^2$ we fix $W(x,y)=(\phi,\phi)$ independently with probability $\e$; otherwise $W(x,y)=(x,y)$.
\end{definition}

In words, any {\em typical} channel $W \in \cW_{Q,\e}$ is {\em almost} the identity function.
It only deviates from the identity function on an $\e$-fraction of input values in expectation, and in such case returns the value $(\phi,\phi)$.
In addition, using Markov's inequality, it follows that with probability at least $1/2$ (over $W \in \cW_{Q,\e}$) the channel $W$ deviates from the identity on at most a $2\e$-fraction of input values.
This implies that with probability at least $1/2$  we have that $\mathcal{R}^{\left(2\e\right)}_{W} \geq \mathcal{R}^{\left(2\e\right)}_{W,1}= 2q$, which is optimal.

In light of Statement~\ref{state:1}, we ask how far $\Rzero$ is from $q$.
First of all we note that for parameters $Q$ and $\e$ in which $Q$ is small with respect to $\e$ (e.g., $\e < 1/Q$) it holds for typical $W \in \cW_{Q,\e}$ that $\Rzero$ is close to $2q$ (which does not support Statement~\ref{state:1}). This follows from the fact that in such channels there are very few input pairs that result in erasures.
Thus, for any $\e>0$, we focus on values of $Q$ which are large and satisfy $Q \geq \Omega(1/\e)$.
Secondly, we remark that finding zero error codes for $W \in \cW_{Q,\e}$ seems challenging as a standard analysis of the natural encoding scheme in which we encode the source information via an erasure code and send the codewords over the channel will not improve on the trivial sum rate $q$ for values of $Q \geq \Omega(1/\e)$.

In what follows we support Statement~\ref{state:1} by showing the existence of channels $W\in\cW_{Q,\e}$ for which on one hand
$\mathcal{R}^{\left(2\e\right)}_{W} = 2q$, while on the other $\Rzeron \le (1-\frac{1}{n})2q$.
In other words, for every fixed $n$ we establish a gap between $\Rzeron$ and $\mathcal{R}^{\left(2\e\right)}_{W}$.
Our results do not have any asymptotic significance since as $n$ grows we approach the trivial bound $\Rzeron \le 2q$.
This is stated formally in Theorem \ref{the:upper} and Corollary~\ref{the:upper} in Section~\ref{sec:upper}.
In Section~\ref{sec:lower}, we study what we view as a natural attempt (that differs from the scheme based on erasure codes discussed above) to show that $\Rzero > q$.

\section{Upper bounds assuming finite block length}
\label{sec:upper}

In this section we present an upper bound on the rate $\Rzeron$ for channels $W$ chosen from the aforementioned distribution $\cW_{Q,\e}$.
For any error value $\e>0$, we study the distribution $\cW_{Q,\e}$ for values of $Q$ which are sufficiently large and satisfy $Q= \Omega(1/\e)$.
\isitout{Posing a lower bound on $Q$ that depends on $\e$ is essential as it is not hard to see that for small $Q$ (say $Q \leq 1/\e$) ``typical'' channels $W$ in the support of $\cW_{Q,\e}$ will have $\Rzero$ which is close to $2q$.}

\begin{theorem}
\label{the:upper} For every integer $n \geq 2$, $\eps \in [0,1]$ and $\gamma >0$, let $Q=2^q$ with $q \geq \max\{\log n,\frac{4}{\gamma}\log \frac{3}{\eps}\}$.   Then with probability at least $3/4$, a random channel $W\in\cW_{Q,\e}$ satisfies
\begin{equation}\label{eq:R0vsReps}
\mathcal{R}^{\left(0\right)}_{W,n} \le 2q\left(1-\frac{1}{n}\right)(1+\gamma).
\end{equation}
Specifically, for $n=2$
\begin{equation*}
\mathcal{R}^{\left(0\right)}_{W,2} \le q(1+\gamma).
\end{equation*}
\end{theorem}

We thus conclude (based on the earlier discussion) that
\begin{corollary}
\label{cor:upper}
For every integer $n \geq 2$, $\eps \in [0,1]$ and $\gamma >0$, let $Q=2^q$ with $q \geq \max\{\log n,\frac{4}{\gamma}\log \frac{3}{\eps}\}$.   Then there exist channels $W\in\cW_{Q,\e}$ such that
$\mathcal{R}^{\left(2\e\right)}_{W} = 2q$ and
$\mathcal{R}^{\left(0\right)}_{W,n} \le 2q\left(1-\frac{1}{n}\right)(1+\gamma).$
\end{corollary}

The proof of Theorem \ref{the:upper} consists of two parts. The first (Proposition \ref{prop:reduction}) reduces the communication rate $\mathcal{R}^{\left(0\right)}_{W,n}$ to a \emph{bipartite independent set} (BPIS) problem in a suitably constructed graph $G_{W,n}$. The second part (Proposition \ref{prop:SizeOfLargestBPIS}) upper bounds the size of the largest BPIS in that graph.
Given the channel $W$, let $G_{W,n}$ be the bipartite graph with vertex set $[Q]^n \cup [Q]^n$ and an edge $(\x,\y)$ if there exists at least one index $i$ s.t. $W(x_i,y_i)=(\phi,\phi)$.

Given a bipartite graph $H=H(X\cup Y,E)$, a BPIS is a pair $(A,B)$, $A \subseteq X$ and $B \subseteq Y$, such that $E(H) \cap (A \times B)=\emptyset$. Here, $E(H)$ is the edge set of $H$. We define the size of the BPIS $(A,B)$ to be $|A||B|$.

%\medskip
%Let $\alpha=\alpha(r,q,n)=\frac{r}{2}+q\left(1-\frac{1}{n}\right)+\frac{1+\log n}{n}$.
\begin{proposition}\label{prop:reduction} Let $W$ be any channel from the support of $\cW_{Q,\e}$.
If $\mathcal{R}^{\left(0\right)}_{W,n} \ge r$, then  $G_{W,n}$ has a BPIS of size at least $2^{rn}-2^{rn/2}((1+2^q)^n-2^{nq})$.
\end{proposition}

\medskip

\begin{proposition}\label{prop:SizeOfLargestBPIS}
Let $W$ be a random channel chosen according to the distribution $\cW_{Q,\e}$, with the corresponding graph $G_{W,n}$. With probability at least $3/4$, the largest BPIS $(A,B)$ in $G_{W,n}$ satisfies
$$\frac{1}{n}\log |A||B| \le q+\log\left(\frac{3}{\eps}\right).$$
\end{proposition}

Before we proceed with the proofs of Propositions \ref{prop:reduction} and \ref{prop:SizeOfLargestBPIS}, we use them to derive Theorem \ref{the:upper}.

\begin{proof} (Theorem \ref{the:upper})
We reinterpret $(1+2^q)^n-2^{qn}$ as a sum of $(n-1)$ binomial terms. One can easily verify that for $q \ge \log n$, the terms form an increasing series, whose sum is then upper bounded by $n2^{q(n-1)}=2^{\log n + q(n-1)}$.
Now suppose that $\Rzeron > r = 2q\left(1-\frac{1}{n}\right)(1+\gamma)$. Then by Proposition \ref{prop:reduction}, $G_{W,n}$ has a BPIS of size
$$s =2^{rn}-2^{\frac{rn}{2}}((1+2^q)^n-2^{nq}) \ge 2^{rn}-2^{\frac{rn}{2}+\log n + q(n-1)}.$$
Plugging in the value of $r$ and rearranging, one arrives at
$$s \ge 2^{2qn\left(1-\frac{1}{n}\right)\left(1+\gamma\right)} - 2^{2qn\left(1-\frac{1}{n}\right)\left(1+\frac{\gamma}{2}+ \frac{\log n}{qn}\right)}.$$
Since $q \ge 4/\gamma$, we have $\frac{\log n}{qn} \le \frac{\gamma}{4}$. Rearranging again we get
$$s \ge 2^{2qn\left(1-\frac{1}{n}\right)\left(1+\frac{3\gamma}{4}\right)}\cdot \left(2^{2qn\left(1-\frac{1}{n}\right)\frac{\gamma}{4}}-1\right).$$
Since $q \ge 4/\gamma$ and $n \ge 2$, the latter is at least $2^{2qn\left(1-\frac{1}{n}\right)\left(1+\frac{3\gamma}{4}\right)}$. Taking the logarithm we arrive at $$\frac{\log s}{n} \ge 2q\left(1-\frac{1}{n}\right)\left(1+\frac{3\gamma}{4}\right) \ge q+\frac{3q\gamma}{4}.$$
By our choice of $q \ge \frac{4}{\gamma}\log \frac{3}{\eps}$, the latter contradicts the upper bound stated in Propositions \ref{prop:SizeOfLargestBPIS}.
\end{proof}

\isitin{The proofs of Propositions~\ref{prop:reduction} and \ref{prop:SizeOfLargestBPIS} are deleted due to space limitations.
Both appear in the full version of this work \cite{LVLE:13}.
}

\isitout{
\begin{proof} (Proposition \ref{prop:reduction})
By the assumption $\mathcal{R}^{\left(0\right)}_{W,n} \ge r$, it follows that there exist sets $X \subseteq [Q]^n$ and $Y \subseteq [Q]^n$, corresponding to the first and second source respectively, such that $\log (|X||Y|)=rn$, and for every $(\x,\y)\in X \times Y$, $D_1 [W^{(n)}(\x,\y)] = \x$ and $D_2 [W^{(n)}(\x,\y)] = \y$. Define the set $X_i \subseteq Q$ to be $X_i = \{ x \in Q: \exists \x  \in X, x_i = x \}$, that is, $X_i$ is the projection of $X$ to the $i^{th}$ block. Similarly define $Y_i$. Our first goal is to upper bound the number of pairs $(\x,\y) \in X \times Y$ that have at least one index $i$ s.t. $W(x_i,y_i)=(\phi,\phi)$.
The key observation can be summarized as follows.
Consider any pair $(x^{(n)},y^{(n)}) \in X \times Y$ that has exactly $t$ indices $i$ for which $W(x_i,y_i)=(\phi,\phi)$ in locations $i_1,\dots,i_t$.
Due to our assumption of correct decoding, it must be the case that for any other pair $(x'^{(n)},y'^{(n)}) \in X \times Y$ that has exactly $t$ indices $i$ for which $W(x'_i,y'_i)=(\phi,\phi)$ in locations $i_1,\dots,i_t$ the projection of $x^{(n)}$ onto indices in the set $[n] \setminus \{i_1,\dots,i_t\}$ must differ from the projection of $x'^{(n)}$ onto indices in the set $[n] \setminus \{i_1,\dots,i_t\}$ (and the same for $y^{(n)}$ and $y'^{(n)}$).
Otherwise $D_1(W_1^{(n)}(x^{(n)},y^{(n)}))=D_1(W_1^{(n)}(x'^{(n)},y'^{(n)}))$.

Hence the total number of different $\x$'s that belong to a pair $(\x,\y) \in X \times Y$ with exactly $t$
failure is at most $\binom{n}{t}\cdot Q^{n-t}=\binom{n}{t}2^{(n-t)q}$. The total number of $\x$'s for $t \geq 1$ is at most
\begin{align}\label{eq:NumOfx's}
\sum_{t=1}^n & \binom{n}{t}2^{(n-t)q}  = (1+2^q)^n-2^{nq}.
\end{align}
Consider the subgraph $H$ of $G_{W,n}$ induced by $(X,Y)$, and let $H'$ be the graph obtained from $H$ by removing every $\x \in X$ that belongs to a pair $(\x,\y)$ with one or more failures. By definition, the graph $H'$ is a BPIS in $G_{W,n}$. Suppose that $|X|=2^{r_1n},|Y|=2^{r_2n}$ for $r_1 \ge r_2$, $r_1+r_2=r$. Then by $(\ref{eq:NumOfx's})$, the size of the BPIS $H'$ is at least
$$\left(2^{r_1n}-((1+2^q)^n-2^{nq})\right)\cdot 2^{r_2n} \ge 2^{rn}-2^{nr/2}((1+2^q)^n-2^{nq}).$$
\end{proof}

\begin{proof} (Proposition \ref{prop:SizeOfLargestBPIS})
We first bound the size of the largest BPIS $(A,B)$ in $G=G_{W,1}$. Let $s=|A|\cdot |B|$. The probability that $G$ has a BPIS of size $s$ is at most
$$2^Q\cdot 2^Q \cdot (1-\eps)^s = 2^{2Q + \log (1-\eps)s} \le 2^{2Q-\eps s}.$$
If $s \geq (2Q+2)/\eps$, then the above probability is smaller than $1/4$. That is, with  probability at least $3/4$, $G$ has no BPIS larger than $s$.

Next we show that the following holds: $s$ is the size of the largest BPIS in $G$ iff $s^n$ is the size of the largest BPIS in $G_{W,n}$.
One direction is trivial: If $(A,B)$ is a BPIS in $G$ of size $s$, then clearly $(A^n,B^n)$ is a BPIS in $G_{W,n}$, and its size is $s^n$.
On the other hand, let $(A',B')$ be a largest BPIS in $G_{W,n}$. Observe that since $A'$ is maximal, then by definition it must be the Cartesian product $A'_1 \times A'_2 \times \dots A'_k$, where $A'_i$ is the projection of $A'$ to the $i^{th}$ coordinate. The same is true for $B'$. Next observe that for all $i$, $(A'_i,B'_i)$ is a BPIS in $G$, or else there will be an edge in $(A',B')$. If the size of $(A',B')$ is at least $s^n$, then at least one of $(A'_i,B'_i)$ satisfies $|A'_i||B'_i| \geq s$.

To conclude, we have established that with probability at least $3/4$, $G_{W,n}$ has no BPIS of size larger than $\left(\frac{2Q+2}{\eps}\right)^n$. In this case,
$$\log \left(\frac{2Q+2}{\eps}\right)^n \leq n\log  \left(\frac{3Q}{\eps}\right) = n(\log 3 + \log Q + \log(1/\eps)).$$
Rearranging, and replacing $\log Q = \log 2^q = q$, we get that the latter equals $n\left(q+\log(3/\eps)\right)$, as required.
\end{proof}
}

\section{$\gamma$-uniform set systems}
\label{sec:lower}

In this section, we tie the existence of a certain natural combinatorial structure to zero error communication schemes. Namely, in Section~\ref{sec:gamma} we define a combinatorial criterion (called the $\gamma$-uniform criterion) on subsets of $[Q]^n\times[Q]^n$ and show that subsets satisfying this criterion yield good zero error encoding schemes for the typical deterministic interference channels $W \in \cW_{Q,\e}$.
We then study upper bounds and lower bounds on the sizes of $\gamma$-uniform sets in Section~\ref{sec:bounds}. Finally we show that the bounds obtained do not resolve the question of whether $\Rzero$ is strictly larger than $q$ (the time sharing bound) but only partially support the conjecture that $\Rzero \simeq q$.

\subsection{$\gamma$-uniform set systems and their connection to $\Rzero$}
\label{sec:gamma}

\begin{definition}
Given $x^{(n)} \in [Q]^n$ and $i \in [n]$, denote by $x_{i}$ the $i$-th coordinate of $x^{(n)}$.
A pair $(x^{(n)},y^{(n)}) \in [Q]^n \times [Q]^n$ is called $\gamma$-uniform if for each pair $(\alpha, \beta) \in [Q]^2$ it holds that
$$
(1-\gamma)\frac{n}{Q^2} \leq | \{i \in [n] \mid (x_i, y_i) = (\alpha, \beta) \} | \leq (1+\gamma)\frac{n}{Q^2}.
$$
In other words, the number of appearances of any pair $(\alpha, \beta)\in [Q]^2$ in $(x^{(n)}, y^{(n)})$ is bounded by $(1\pm \gamma)\frac{n}{Q^2}$; i.e., the {\em type} of $(x^{(n)},y^{(n)})$ is $\gamma$-far from being uniform (under the $\|\cdot\|_{\infty}$ norm).
Similarly, the subsets $A \subseteq [Q]^n, B \subseteq [Q]^n$ are called  $\gamma$-uniform if for any $x^{(n)} \in A, y^{(n)} \in B$, $(x^{(n)},y^{(n)})$ are  $\gamma$-uniform.
\end{definition}

%\isitin{The proof of the following theorem is omitted due to space limitations and appears in \cite{LVLE:13}.}

The following theorem ties the existence of $\gamma$-uniform set systems to good zero error codes for typical channels $W$ in $\cW_{Q,\e}$. Roughly speaking, given a $\gamma$-uniform pair $A$ and $B$ one can construct a zero error code for $W$ by taking large subsets $A'$ of $A$ and $B'$ of $B$ with {\em large} minimum distance.
Here the term {\em large} depends on $\e$ and $\gamma$.
\isitin{The details of our construction and the proof of the following theorem are omitted and appear in the full version of this work \cite{LVLE:13}.}

\begin{theorem}
\label{the:code_rate_corollary}
Let $Q=2^q$. Let $A \subseteq [Q]^n, B \subseteq [Q]^n$ be  $\gamma$-uniform with
$|A| \geq Q^{n(1-\delta_1)}$ and $|B| \geq Q^{n(1-\delta_2)}.$
Let $\delta >0$ be arbitrarily small.
Consider a channel $W$ chosen from the distribution $\cW_{Q,\e}$.
With probability at least $1/2$ it holds that:
\begin{equation}\label{eq:code_rate_corollary_eq}
\mathcal{R}^{\left(0\right)}_{W,n} \geq 2q\left(1-\left(\frac{\delta_1+\delta_2}{2}+2(1+\gamma)\e+\delta\right)\right)-2.
\end{equation}
\end{theorem}

\isitout{
To prove Theorem~\ref{the:code_rate_corollary} we will introduce an additional combinatorial criterion on sets. We refer to the additional criterion as the $(d,\e)$-diversity criterion.

\begin{definition}
A pair $(x^{(n)},y^{(n)}) \in [Q]^n \times [Q]^n$ is called $(d,\e)$-diverse if for each index set $I \subseteq [n]$ of size $dn$ it holds that $|\{ (x_j, y_j) \mid j \in I \}| > \e Q^2$. Similarly, the subsets $A \subseteq [Q]^n, B \subseteq [Q]^n$ are called $(d,\e)$-diverse if for any $x^{(n)} \in A, y^{(n)} \in B$, $(x^{(n)},y^{(n)})$ are  $(d,\e)$-diverse.
\end{definition}

We first connect $(d,\e)$-diverse set systems to good zero error codes for channels in $\cW_{Q,\e}$ (via Theorem~\ref{the:code_rate_theorem} below). We then turn to prove Theorem~\ref{the:code_rate_corollary}.

\begin{theorem}
\label{the:code_rate_theorem}
Let $Q=2^q$. Let $A \subseteq [Q]^n, B \subseteq [Q]^n$ be  $(d,2\epsilon)$-diverse with
$|A| \geq Q^{n(1-\delta_1)}$ and $|B| \geq Q^{n(1-\delta_2)}.$
Consider a channel $W$ chosen from the distribution $\cW_{Q,\e}$.
With probability at least $1/2$ it holds that
\begin{equation*}
\mathcal{R}^{\left(0\right)}_{W,n} \geq 2q\left(1-\left(\frac{\delta_1+\delta_2}{2}+d\right)\right)-2.
\end{equation*}
\end{theorem}

To prove Theorem~\ref{the:code_rate_theorem}, we first require the following lemma that follows from a standard packing argument.

\begin{lemma}
\label{lem:big_subset}
Let $A \subseteq [Q]^n$. Then for any $d \in [0,1]$ there exists $A' \subseteq A$ such that $|A'| > \frac{|A|}{2^n Q^{dn}}$ and for any $x^{(n)},x^{\prime(n)} \in A'$ $h(x^{(n)},x^{\prime(n)}) > dn$, where $h : [Q]^n \times [Q]^n \to \{0, 1, \ldots, n\}$ is the Hamming distance function.
\end{lemma}

\begin{proof} (Lemma~\ref{lem:big_subset})
Consider a graph $G = (V, E)$ where the vertices are elements of $A$, and there is an edge between two vertices $x^{(n)}, x^{\prime(n)}$ if and only if $h(x^{(n)}, x^{\prime(n)}) \leq dn$. The maximal degree of a vertex in this graph is  $\binom{n}{dn} Q^{dn}$. Thus, the size of the independent set in $G$ is at least $\frac{|A|}{\binom{n}{dn} Q^{dn}} > \frac{|A|}{2^n Q^{dn}}$, and the vertices of this independent set satisfy the conditions on $A'$ in the lemma.
\end{proof}

\begin{proof} (Theorem~\ref{the:code_rate_theorem})
We first note that with probability at least $1/2$, $W$ chosen at random
from $\cW_{Q,\e}$ has at most $2\e Q^2$ distinct values $(x,y) \in [Q]^2$ s.t. $W(x,y)=(\phi,\phi)$. This follows from the Markov inequality.

Let $A' \subseteq A, B'  \subseteq B$ be the subsets whose existence is guaranteed by Lemma~\ref{lem:big_subset}. We claim that communication with block length $n$ over $W$ is successful on input $(x^{(n)}, y^{(n)}) \in A' \times B'$.
Let us assume the contrary. Namely, that there exist $x^{(n)}, x^{\prime(n)} \in A'$,  $y^{(n)}, y^{\prime(n)} \in B'$ such that: (a) $x^{(n)} \neq x^{\prime(n)}$ and $W^{(n)}_1(x^{(n)}, y^{(n)}) = W^{(n)}_1(x^{\prime(n)}, y^{\prime(n)})$ or (b)  $y^{(n)} \neq y^{\prime(n)}$ and $W^{(n)}_2(x^{(n)}, y^{(n)}) = W^{(n)}_2(x^{\prime(n)}, y^{\prime(n)})$. Without loss of generality, consider option (a). Since $h(x^{(n)}, x^{\prime(n)}) > dn$, there exists an index set $I \subset [n]$ of size $dn$ such that for any $i \in I$, $x_{i} \neq x^{\prime(n)}_{i}$. By our assumption $W^{(n)}_1(x^{(n)}, y^{(n)}) = W^{(n)}_1(x^{\prime(n)}, y^{\prime(n)})$, this means that for each $i \in I$, $\left(W^{(n)}_1(x^{(n)}, y^{(n)})\right)_{i} = \left(W^{(n)}_1(x^{\prime(n)}, y^{\prime(n)})\right)_{i} = \phi$. Now, as $A, B$ are $(d, 2\e)$-diverse, so are $A', B'$.  Thus,  $|\{ (x_{i}, y_{i}) \mid i \in I \}| > 2\e Q^2$ in contradiction to the fact that for $W$, $|\{(\alpha, \beta) \in [Q]^2 | W_1(\alpha, \beta) = \phi\}| \leq 2\e Q^2$. Finally, note that by using $A', B'$ we can achieve a rate of $\left(\frac{\log |A'|}{n}, \frac{\log |B'|}{n}\right)$. By Lemma~\ref{lem:big_subset}, this rate is lower bounded by:
$$
\left(\frac{\log \frac{|A|}{2^n Q^{dn}}}{n}, \frac{\log \frac{|B|}{2^n Q^{dn}}}{n}\right) \geq  \left(\frac{\log \frac{Q^{n(1-\delta_1)}}{2^n Q^{dn}}}{n}, \frac{\log \frac{Q^{n(1-\delta_2)}}{2^n Q^{dn}}}{n}\right)
$$
which is equal to:
$$
q(1-(\delta_1+d)) - 1, q(1-(\delta_2+d)) - 1
$$
and this rate yields the asserted bound on $\mathcal{R}^{\left(0\right)}_{W,n}$.
\end{proof}
}

\isitout{
We now tie $\gamma$-uniform set systems to $(d,\e)$-diverse systems.

\begin{lemma}
\label{gamma_to_balance}
 If $d > (1+\gamma)\e$ then: $(A,B)$  is $\gamma$-uniform $\Rightarrow$ $(A,B)$  is $(d,\e)$-diverse.
\end{lemma}

\begin{proof}
Let  $(x^{(n)},y^{(n)}) \in [Q]^n \times [Q]^n$ be $\gamma$-uniform and $I \subseteq [n]$ some index set of size $dn$. A pair $(\alpha, \beta) \in [Q]^2$ can appear at most $(1+\gamma)\frac{n}{Q^2}$ times in $(x^{(n)}, y^{(n)})$. Particularly, $|\{i \mid (x_i, y_i) = (\alpha, \beta), i \in I\}| \leq (1+\gamma)\frac{n}{Q^2}$. This means that $|\{(x_i, y_i) \mid i \in I\}| \geq \frac{dn}{(1+\gamma)\frac{n}{Q^2}} = \frac{d}{1 + \gamma} Q^2$. But if $d > (1+\gamma)\e$ then $\frac{d}{1 + \gamma} Q^2 > \e Q^2$, and thus $(x^{(n)}, y^{(n)})$ is $(d, \e)$-diverse.
\end{proof}

Finally, we conclude with the proof of Theorem~\ref{the:code_rate_corollary}:
\begin{proof} (Theorem~\ref{the:code_rate_corollary})
The proof of Theorem~\ref{the:code_rate_corollary} follows directly by combining Lemma~\ref{gamma_to_balance} with Theorem~\ref{the:code_rate_theorem}.
\end{proof}
}

\subsection{Upper and Lower bounds on $\gamma$-uniform set systems}
\label{sec:bounds}

The previous section presented a scheme  to construct codes for channels $W$ chosen at random from the distribution $\cW_{Q,\e}$ based on $\gamma$-uniform set systems. We now attempt to better understand the parameters for which such set systems exist. The following lemmas present both upper and lower bounds on the size of $\gamma$-uniform set systems. We then elaborate on the implication of our bounds on Theorem~\ref{the:code_rate_corollary}.
\isitin{The proofs of our claims are omitted due to space limitations and can be found in \cite{LVLE:13}.}

%\subsubsection{Lower bounds}

\begin{lemma}
\label{lem:lower_gamma_2}
Let $0 < \gamma < 2$. Let $n$ be divisible by 4. There exists a $\gamma$-uniform pair $A \subseteq [2]^n, B \subseteq [2]^n$ such that:
\begin{equation*}
|A||B| > \frac{2}{n (n+1)} \cdot 2^{\left(1 + H\left(\frac{\gamma}{4}\right)\right)n}.
\end{equation*}
where $H$ is the binary entropy function.
\end{lemma}

\isitout{
\begin{proof}
Define $a^{(n)}, \bar{a}^{(n)} \in [2]^n$ to be:
\begin{eqnarray*}
a_i =
\begin{cases}
1 & \text{ if } i \leq \frac{n}{2} \\
2 & \text{ otherwise } \\
\end{cases} \\
\bar{a}_i =
\begin{cases}
1 & \text{ if } i > \frac{n}{2} \\
2 & \text{ otherwise } \\
\end{cases} \\
\end{eqnarray*}
Namely $a^{(n)}$ has the form $1^{\frac{n}{2}}2^{\frac{n}{2}}$ and $\bar{a}^{(n)}$ is its bitwise inverse and has the form $2^{\frac{n}{2}}1^{\frac{n}{2}}$.
Consider the sets  $A' \subseteq [2]^n, B \subseteq [2]^n$
where $A' = \{ a^{(n)}, \bar{a}^{(n)}\}$ and $B$ is the maximal set such that $(A',B)$ is zero-uniform.
It is easy to see that
\begin{equation*}
|A'||B| = 2 \cdot \binom{\frac{n}{2}}{\frac{n}{4}}^2
\end{equation*}
since such $B$ can be obtained by selecting $y^{(n)}$-s that have exactly $\frac{n}{4}$ ones in the range $y_1 \ldots y_{\frac{n}{2}}$ and exactly $\frac{n}{4}$ ones in the range $y_{\frac{n}{2}+1} \ldots y_{n}$.

Let $A = \{ x^{(n)} \mid \min(d(x^{(n)}, a^{(n)}), d(x^{(n)}, \bar{a}^{(n)})) \leq \frac{\gamma n}{4}\}$, where $d$ is the Hamming distance.
$(A, B)$ is $\gamma$-uniform, and since $\gamma < 2$:
\begin{equation*}
|A||B| = \left( 2 \sum_{i \leq \frac{\gamma n}{4}} \binom{n}{i} \right) \cdot \binom{\frac{n}{2}}{\frac{n}{4}}^2
\end{equation*}
Using the lower bounds (see, e.g., \cite{KR00} for the first one):
\begin{eqnarray*}
\binom{\frac{n}{2}}{\frac{n}{4}}^2 \geq \frac{2^n}{n} \\
\sum_{i \leq \frac{\gamma n}{4}} \binom{n}{i} > \binom{n}{\frac{\gamma n}{4}} \geq \frac{2^{H(\frac{\gamma}{4})n}}{n+1}
\end{eqnarray*}
We obtain the bound of Lemma~\ref{lem:lower_gamma_2}:
\begin{equation*}
|A||B| > \frac{2}{n (n+1)} \cdot 2^{(1 + H(\frac{\gamma}{4}))n}.
\end{equation*}

\end{proof}
}

\begin{lemma}
\label{lem:lower_gamma_Q}
Let  $Q \geq 3$ and $n \geq {Q^3}$ such that $n$ is divisible by $Q^2$.
For $\gamma \leq Q^2$ there exists a $\gamma$-uniform pair $A \subseteq [Q]^n, B \subseteq [Q]^n$ such that:
\begin{equation*}
|A||B| \geq \left(\frac{1}{n}\right)^{\frac{Q^2}{2}} \cdot 2^{H\left(\frac{\gamma}{Q^2}\right) n} \cdot (Q-1)^{\frac{\gamma n}{Q^2}} \cdot Q^n.
\end{equation*}
\end{lemma}

\isitout{
\begin{proof}
The technique of this proof is similar to the one used in Lemma~\ref{lem:lower_gamma_2}.
Define $a^{(n)} \in [Q]^n$ to be:
$$
a_i = \left\lceil \frac{Qi}{n} \right\rceil  \text { for all } 1 \leq i \leq n
$$
Namely $a^{(n)}$ has the form $1^{\frac{n}{Q}}2^{\frac{n}{Q}} \ldots Q^{\frac{n}{Q}}$.
Consider the sets  $A' \subseteq [Q]^n, B \subseteq [Q]^n$
where $A' = \{ a^{(n)} \}$ and $B$ is the maximal set such that $(A', B)$ is zero-uniform.
It holds that
$$
|A'||B| = |B| = \left(\prod_{i=1}^{Q}\binom{i\cdot\frac{n}{Q^2}}{\frac{n}{Q^2}}\right)^Q
$$

Let $A = \{ x^{(n)} \mid d(x^{(n)}, a^{(n)}) \leq \frac{\gamma n}{Q^2}\}$, where $d$ is the Hamming distance.
$(A, B)$ is $\gamma$-uniform, and
\begin{equation*}
|A||B| = \left(\sum_{i \leq \frac{\gamma n}{Q^2}} \left( \binom{n}{i} (Q-1)^{i} \right) \right) \cdot \left(\prod_{i=1}^{Q}\binom{i\cdot\frac{n}{Q^2}}{\frac{n}{Q^2}}\right)^Q
\end{equation*}

Notice that:
\begin{equation*}
\prod_{i=1}^{Q}\binom{i\cdot\frac{n}{Q^2}}{\frac{n}{Q^2}} = \frac{\left(\frac{n}{Q}\right)!}{\left(\left(\frac{n}{Q^2}\right)!\right)^Q}
\end{equation*}

To evaluate this expression we can use the following bounds on the factorial that result from Stirling's formula \cite{ROB55}:
\begin{equation*}
\sqrt{2 \pi n}\left(\frac{n}{e}\right)^n e^{\frac{1}{12n+1}} < n! < \sqrt{2 \pi n}\left(\frac{n}{e}\right)^n e^{\frac{1}{12n}}
\end{equation*}

Which means that:
\begin{eqnarray*}
\left(\frac{n}{Q}\right)! > \sqrt{2 \pi \left(\frac{n}{Q}\right)}\left(\frac{\left(\frac{n}{Q}\right)}{e}\right)^{\left(\frac{n}{Q}\right)} e^{\frac{1}{12\left(\frac{n}{Q}\right)+1}} \\
\left(\frac{n}{Q^2}\right)! < \sqrt{2 \pi \left(\frac{n}{Q^2}\right)}\left(\frac{\left(\frac{n}{Q^2}\right)}{e}\right)^{\left(\frac{n}{Q^2}\right)} e^{\frac{1}{12\left(\frac{n}{Q^2}\right)}}
\end{eqnarray*}

This provides the following upper bound:
\begin{eqnarray*}
 \frac{\left(\frac{n}{Q}\right)!}{\left(\left(\frac{n}{Q^2}\right)!\right)^Q} > \frac{Q^{Q-\frac{1}{2}}}{(2 \pi n)^{\frac{Q}{2} - \frac{1}{2}}} \cdot Q^{\frac{n}{Q}} \cdot e^{\frac{Q}{12n + Q} - \frac{Q^3}{12n}}  >  \\
\frac{Q^{Q-\frac{1}{2}}}{(2 \pi n)^{\frac{Q}{2} - \frac{1}{2}}} \cdot Q^{\frac{n}{Q}} \cdot e^{- \frac{Q^3}{12n}}
\end{eqnarray*}

Which for $n \geq Q^3$ is greater than:
\begin{eqnarray*}
\frac{Q^{Q-\frac{1}{2}}}{(2 \pi n)^{\frac{Q}{2} - \frac{1}{2}}} \cdot Q^{\frac{n}{Q}} \cdot e^{- \frac{1}{12}} > e^{- \frac{1}{12}} \cdot \left( \frac{Q^2}{2 \pi n} \right)^{\frac{Q}{2}-\frac{1}{2}}\cdot Q^{\frac{n}{Q}}
\end{eqnarray*}

And thus:
\begin{eqnarray*}
|A||B|  = \left(\sum_{i \leq \frac{\gamma n}{Q^2}} \left( \binom{n}{i} (Q-1)^{i} \right) \right) \cdot \left(\prod_{i=1}^{Q}\binom{i\cdot\frac{n}{Q^2}}{\frac{n}{Q^2}}\right)^Q \\
 >  \binom{n}{\frac{\gamma n}{Q^2}} (Q-1)^{\frac{\gamma n}{Q^2}} \cdot e^{- \frac{Q}{12}} \cdot \left( \frac{Q^2}{2 \pi n} \right)^{\frac{Q(Q-1)}{2}}\cdot Q^n \\
>  \frac{2^{H\left(\frac{\gamma}{Q^2}\right) n}}{n+1} \cdot e^{- \frac{Q}{12}} \cdot \left( \frac{Q^2}{2 \pi n} \right)^{\frac{Q(Q-1)}{2}} \cdot (Q-1)^{\frac{\gamma n}{Q^2}} \cdot Q^n \\
 >  \left(\frac{1}{n}\right)^{\frac{Q^2}{2}} \cdot 2^{H\left(\frac{\gamma}{Q^2}\right) n} \cdot (Q-1)^{\frac{\gamma n}{Q^2}} \cdot Q^n
\end{eqnarray*}

\end{proof}
}

%\subsubsection{Upper bounds}

\begin{lemma}
\label{lem:upper_gamma_Q}
If $A \subseteq [Q]^n, B \subseteq [Q]^n$ are $\gamma$-uniform with $Q = 2^q$, $\gamma \leq 2$, then
$$
|A||B| \leq 2^{nq\left(1+\frac{\gamma}{2}+H\left(\frac{\gamma}{2}\right)\right)}.
$$
\end{lemma}

\isitout{
In order to prove Lemma~\ref{lem:upper_gamma_Q} we need a few other lemmas. The first lemma is a result of Sgall \cite{SGALL98}, which in our terms states the following:
\vspace{2mm}
\begin{lemma}\cite[Corollary 3.5]{SGALL98}
\label{lem:Sgall}
Let $A \subseteq [2]^n, B \subseteq [2]^n$ be  $\gamma$-uniform ($\gamma \leq 2$). Then:
\begin{equation*}
|A||B| \leq \binom{n}{\frac{\gamma n}{2}} 2^{n(1 + \frac{\gamma }{2})} \leq 2^{n(1 + \frac{\gamma}{2} + H(\frac{\gamma}{2}))}
\end{equation*}
\end{lemma}

The next lemma provides a framework for creating a reduction from any $Q=2^q$ to $Q = 2$. Our proof essentially uses the binary representation of elements in $[Q]$ but is presented in a general manner to support a similar (although slightly more complicated) proof that can be used if needed for any value of $Q$ (not necessarily of size $2^q$).
\vspace{2mm}
\begin{lemma}
\label{lem:maintain_even_Q}
Let $Q =2^q$ and let $f, g : [Q] \to [2]^{q}$ be any functions such that for any $1 \leq i \leq q$ it holds that
\begin{eqnarray*}
|\{ x \in [Q] \mid f(x)_i = 1\}| = \frac{Q}{2} \\
|\{ y \in [Q] \mid g(y)_i = 1\}| = \frac{Q}{2}
\end{eqnarray*}
Define functions $f^{(n)}, g^{(n)} : [Q]^n \to [2]^{qn}$ to be
\begin{eqnarray*}
f^{(n)} (x^{(n)}) = f (x_1) f (x_2) \ldots f (x_n) \\
g^{(n)} (y^{(n)}) = g (y_1) g (y_2) \ldots g (y_n)
\end{eqnarray*}

Let $A \subseteq [Q]^n, B \subseteq [Q]^n$ and define subsets $A' \subseteq [2]^{q n}, B' \subseteq [2]^{q n}$ to be
\begin{eqnarray*}
A' = \{ f^{(n)} (x^{(n)}) \mid  x^{(n)} \in A\} \\
B' = \{ g^{(n)} (y^{(n)}) \mid  y^{(n)} \in B\}
\end{eqnarray*}

If $A, B$ are   $\gamma$-uniform then $A', B'$ are $\gamma$-uniform.
\end{lemma}

\begin{proof}
Let $f, g$ be functions as in the lemma statement. For any $1 \leq i \leq q$ and some $(\alpha, \beta) \in [2]^2$ let us define $S_{f, g, i} (\alpha, \beta) = \{ (x, y) \in [Q]^2 \mid (f(x)_i, g(y)_i) = (\alpha, \beta) \}$. Namely, $S_{f, g, i} (\alpha, \beta)$ is the set of pairs in $[Q]^2$ that "generate" the pair $(\alpha, \beta)$ at position $i$, with respect to functions $f$ and $g$. Note that the restrictions on $f,g$ in the lemma imply that $|S_{f, g, i} (\alpha, \beta)| = \frac{Q^2}{4}$.

Let $A, B$ be $\gamma$-uniform, $(x^{(n)}, y^{(n)}) \in A \times B$, and consider the number of times a pair $(\alpha, \beta) \in [2]^2$ appears in $\left(f^{(n)} (x^{(n)}), g^{(n)} (y^{(n)})\right)$. This number is equal to
\begin{equation}
\label{even_Q_proof1}
\sum_{i=1}^{q} \sum_{(x,y) \in S_{f, g, i} (\alpha, \beta)} C_{(x,y)}(x^{(n)}, y^{(n)})
\end{equation}
where
\begin{equation*}
C_{(x,y)}(x^{(n)}, y^{(n)})  = | \{ k \in [n] \mid (x_k, y_k) = (x,y) \} |
\end{equation*}
Since  $(x^{(n)}, y^{(n)})$ is $\gamma$-uniform,  $C_{(x,y)}(x^{(n)}, y^{(n)})$ is bounded by $(1 \pm \gamma) \frac{n}{Q^2}$. Thus, Expression~\ref{even_Q_proof1} is bounded by $q \cdot \frac{Q^2}{4} \cdot (1 \pm \gamma) \frac{n}{Q^2} =(1 \pm \gamma) \frac{q n}{4}$. Hence  $A',B'$ are $\gamma$-uniform.
\end{proof}

Now we can prove Lemma~\ref{lem:upper_gamma_Q}.
\begin{proof} (Lemma~\ref{lem:upper_gamma_Q})
Let $A \subseteq [Q]^n, B \subseteq [Q]^n$ be $\gamma$-uniform.
For $Q = 2$, the upper bound is provided by Lemma~\ref{lem:Sgall}.
The upper bounds for $Q=2^q$ is obtained by constructing $A', B'$ via Lemma~\ref{lem:maintain_even_Q}, and then using Lemma~\ref{lem:Sgall} on $A', B'$.
All that remains is to prove the existence of injective $f, g$ that satisfy the conditions of Lemma~\ref{lem:maintain_even_Q}.
Indeed, define $f$ greedily in the following way: let $P = \{P_1, P_2, \ldots P_\frac{Q}{2}\}$ be any partition of $[Q]$ such that  $|P_i| = 2$ (for any $i$). For each $P_i = (v_1, v_2)$ set $f(v_1)$ to any previously unused value  $w \in \{1, 2\}^q$, and set $f(v_2)$ to the bitwise inverse of $w$ (1-s replaced with 2-s and 2-s replaced with 1-s). Also set $g = f$.
It is easy to see that $f, g$ satisfy the conditions of Lemma~\ref{lem:maintain_even_Q}  and are injective. (Note that for the case $Q = 2^q$, described above, any bijection $f$ satisfies the conditions of Lemma 9, in particular the binary representation. The described algorithm  is more generic, and can be used for any even $Q$).
\end{proof}
}

\subsection{Implications on $\Rzero$}
\label{sec:implications}

In this section we show that the upper and lower bounds presented above combined with Theorem~\ref{the:code_rate_corollary} do not resolve the question whether $\Rzero$ is greater than $q$.
Namely, we show that the lower bound on $\gamma$-uniform set systems does not imply that $\Rzero>q$.
In addition, to put our result in context, we also show that an optimistic assumption that there exist $\gamma$-uniform set systems that match the upper bound of the previous section does indeed imply that $\Rzero>q$, however our upper bound may be loose and such set systems are not known to exist.
All in all, even though we cannot conclude any bounds on the value of $\Rzero$ for our channels $W \in \cW_{Q,\e}$, we believe that the concept of $\gamma$-uniform set systems is an interesting one and that a better understanding of bounds for such systems may lend insight  into the value of $\Rzero$.
\isitin{The proofs of our claims below are omitted and appear in \cite{LVLE:13}.}

\begin{claim}
Let $A \subseteq [Q]^n, B \subseteq [Q]^n$ be $\gamma$-uniform with $|A||B|$ equal to the lower bound in Lemma~\ref{lem:lower_gamma_Q}.
If $\e > \frac{1}{2Q^2}$, then the RHS of equation~\ref{eq:code_rate_corollary_eq} is no larger than $q - 1$.
\end{claim}

\isitout{
\begin{proof}
We will prove a stronger statement by using a lower bound of $2^{nq(1+\gamma/Q^2+1/q)}$
which is larger than the lower bound in Lemma~\ref{lem:lower_gamma_Q}.
In order to satisfy the conditions of Theorem~\ref{the:code_rate_corollary}:
\begin{eqnarray*}
|A||B| = 2^{nq(1+\gamma/Q^2+1/q)} & \geq &
2^{nq(1-\delta_1) + nq (1-\delta_2)} \\
& = & 2^{nq(2-\delta_1-\delta_2)}
\end{eqnarray*}
Thus:
$$
1+\gamma/Q^2+1/q \geq 2-\delta_1-\delta_2
$$
Which means that:
$$
\frac{\delta_1+\delta_2}{2} \geq \frac{1}{2} - \frac{\gamma}{2Q^2}-\frac{1}{2q}
$$
Thus we can bind the RHS of equation~\ref{eq:code_rate_corollary_eq} from above:
\begin{eqnarray*}
2q\left(1-\left(\frac{\delta_1+\delta_2}{2}+2(1+\gamma)\e+\delta\right)\right)-2  \leq \\ 2q\left(1-\left(\frac{1}{2} - \frac{\gamma}{2Q^2}-\frac{1}{2q} +2(1+\gamma)\e+\delta\right)\right)-2 = \\
2q\left(\frac{1}{2} + \frac{\gamma}{2Q^2} - 2(1+\gamma)\e - \delta\right)-1
\end{eqnarray*}
And since $\e > \frac{1}{2Q^2}$ the whole expression is less than $q - 1$.
\end{proof}
}

\begin{claim}
Let $A \subseteq [Q]^n, B \subseteq [Q]^n$ be $\gamma$-uniform with $|A||B|$ equal to the upper bound in Lemma~\ref{lem:upper_gamma_Q}.
Then the RHS of equation~\ref{eq:code_rate_corollary_eq} is equal to
$$
2q\left(\frac{1}{2}+\frac{\gamma}{4}+\frac{H(\frac{\gamma}{2})}{2}-2\e(1+\gamma)-\delta\right)-2,
$$
which is greater than $q$ for $\e<\frac{\gamma+2H(\gamma/2)}{8(1+\gamma)}$ and sufficiently large $Q$.
\end{claim}

\isitout{
\begin{proof}
According to the conditions of Theorem~\ref{the:code_rate_corollary}:
\begin{eqnarray*}
|A||B| =  2^{nq\left(1+\frac{\gamma}{2}+H\left(\frac{\gamma}{2}\right)\right)} & = & 2^{nq(1-\delta_1) + nq (1-\delta_2)} \\
& = & 2^{nq(2-\delta_1-\delta_2)}
\end{eqnarray*}
Thus:
$$
1+\frac{\gamma}{2}+H\left(\frac{\gamma}{2}\right) = 2-\delta_1-\delta_2
$$
Which means that:
$$
\frac{\delta_1+\delta_2}{2} = \frac{1}{2} - \frac{\gamma}{4}-\frac{H\left(\frac{\gamma}{2}\right)}{2}
$$
Thus the RHS of equation~\ref{eq:code_rate_corollary_eq} is:
\begin{eqnarray*}
2q\left(1-\left(\frac{\delta_1+\delta_2}{2}+2(1+\gamma)\e+\delta\right)\right)-2  = \\ 2q\left(1-\left(\frac{1}{2} - \frac{\gamma}{4}-\frac{H\left(\frac{\gamma}{2}\right)}{2}+2(1+\gamma)\e+\delta\right)\right)-2 = \\
2q\left(\frac{1}{2} + \frac{\gamma}{4} + \frac{H\left(\frac{\gamma}{2}\right)}{2} - 2(1+\gamma)\e - \delta\right)-2
\end{eqnarray*}
To see that this is may better guarantee than the time sharing scheme, consider, e.g., the case $\gamma = 1, \e = \frac{1}{8}$. The expression above will evaluate to $1.5q - 2\delta q - 2>q$ for sufficiently large values of $q$ (as $\delta>0$ is arbitrarily small).
\end{proof}
}

\isitout{
\begin{claim}
Assume the existence of $\gamma$-uniform sets $A \subseteq [Q]^n, B \subseteq [Q]^n$  such that $|A||B| \geq Q^{n(1+f(\gamma))}$.
Then if for some constant $c$  it holds that $f(\gamma) > 4\e(1+\gamma) + \frac{2}{q^2} + c$, we have by Theorem~\ref{the:code_rate_corollary} a scheme that improves on the time sharing scheme.
\end{claim}

\begin{proof}
According to the conditions of Theorem~\ref{the:code_rate_corollary}:
\begin{equation*}
|A||B| \geq  Q^{n(1+f(\gamma))}  =  Q^{n(1-\delta_1)}Q^{n(1-\delta_2)}
\end{equation*}
Thus:
$$
1+f(\gamma) =  2-\delta_1-\delta_2
$$
Which means that:
$$
\frac{\delta_1+\delta_2}{2} = \frac{1}{2}-\frac{f(\gamma)}{2}
$$
Thus the RHS of equation~\ref{eq:code_rate_corollary_eq} is:
\begin{eqnarray*}
2q\left(1-\left(\frac{\delta_1+\delta_2}{2}+2(1+\gamma)\e+\delta\right)\right)-2  = \\ 2q\left(1-\left(\frac{1}{2} - \frac{f(\gamma)}{2}+2(1+\gamma)\e+\delta\right)\right)-2 = \\
2q\left(\frac{1}{2} + \frac{f(\gamma)}{2}-2(1+\gamma)\e-\delta\right)-2
\end{eqnarray*}
If $f(\gamma) > 4\e(1+\gamma) + \frac{2}{q^2} + c$ then the RHS of equation~\ref{eq:code_rate_corollary_eq} is bounded from below by $q(1+c-2\delta)$. Since we can select any $\delta > 0$, this improves on the time sharing scheme.
\end{proof}

% Another interesting problem would be to determine whether the existence of a "better than time-sharing" scheme would imply the existence of large  $(d, \e)$-diverse  and $\gamma$-uniform subsets.
%Finally, we note that the use of the above framework to obtain "better than time-sharing" schemes is inherently restricted to certain values of $q$ and $\e$. For example, it can not be used for $q = 1$ or $\e \geq \frac{1}{2}$. The latter is based on the observation that for any $(d, \e)$-diverse pair it holds that $d \geq \epsilon - \frac{1}{Q^2} + \frac{1}{n}$.
}

\section{Conclusion and open problems}
\label{sec:conc}

Motivated by similar questions in network coding, we address the potential gap between $\Rzero$ and $\Reps$ in the context of 2-source/2-terminal deterministic interference channels. In Statement~\ref{state:1} we conjecture that there exist channels $W$ for which $\Rzero \leq \Reps/2 + \delta$ (and more generally for the $k$-source/$k$-terminal case that $\Rzero \leq \Reps/k + \delta$). Studying the channels that result from the distribution $\cW_{Q,\e}$, we support Statement~\ref{state:1} by presenting upper bounds on $\Rzeron$ (which take into account the block length $n$) and by studying the limitations of a natural encoding scheme based on $\gamma$-uniform set systems. We view our posing of Statement~\ref{state:1}, our upper bounds, and the study of $\gamma$-uniform set systems as the main contributions of this work.
Whether Statement~\ref{state:1} is true or not remains an interesting open subject for future research.

%Theorem~\ref{the:code_rate_theorem} and Theorem~\ref{the:code_rate_corollary} provide a framework which can be used to look for schemes that yield a better 0-error sum capacity than time-sharing.
%Such a scheme would exist if there exist subsets  $A \subseteq [Q]^n, B \subseteq [Q]^n$ that are $\gamma$-uniform or $(d, \e)$-diverse and large enough. Unfortunately, our current lower and upper bounds as presented in Lemma~\ref{lem:lower_gamma_2}, \ref{lem:lower_gamma_Q}, \ref{lem:upper_gamma_Q}, can neither
%corroborate nor rule out the existence of these subsets.

{\small
\bibliographystyle{amsplain}
\bibliography{researchbib,MikeBib}
}

\end{document}